\documentclass[a4paper,12pt]{article}
\usepackage{amsmath, amssymb}
\usepackage{amsthm}
\usepackage[OT1]{fontenc}
\usepackage{url}
\usepackage[margin=1in]{geometry}
\DeclareMathOperator*{\argmin}{arg\,min}

\newtheorem{theorem}{Theorem}

\newtheorem{lemma}[theorem]{Lemma}

\newtheorem{remark}[theorem]{Remark}

\newtheorem{assumption}[theorem]{Assumption}

\numberwithin{equation}{section}

\title{  Dynamic Convex Duality in Constrained Utility Maximization}
\date{}
\author{Yusong Li\footnote{Department of Mathematics, Imperial College, London SW7 2BZ, UK.  Email: y.li11@imperial.ac.uk}\; and Harry Zheng\footnote{Corresponding Author. Department of Mathematics, Imperial College, London SW7 2BZ, UK.  Email: h.zheng@imperial.ac.uk}}

\begin{document}
\maketitle



\begin{abstract}
In this paper, we study a constrained utility maximization problem following the convex duality approach. After formulating the primal and dual problems, we construct the necessary and sufficient conditions for both the primal and dual problems in terms of FBSDEs plus additional conditions. Such formulation then allows us to explicitly characterize the primal optimal control as a function of the adjoint process coming from the dual FBSDEs in a dynamic fashion and vice versa. Moreover, we also find that the optimal primal wealth process coincides with the adjoint process of the dual problem and vice versa.  Finally we solve three constrained utility maximization problems, which contrasts the simplicity of the duality approach we propose and the technical complexity of solving the primal problems directly.
\end{abstract}

\noindent\textbf{Keywords}: convex duality, primal and dual FBSDEs, utility maximization, convex portfolio constraints

\bigskip
\noindent\textbf{AMS MSC2010}: 91G80,\ 93E20,\ 49N05,\ 49N15

\section{Introduction}
One of the most commonly studied problems in mathematical economics is the optimal consumption/investment problem. Such problems have as their goal of constructing the investment strategy that maximizes the agent's expected utility of the wealth at the end of the planning horizon. Here we assume that trading strategies take values in a closed convex set which is general enough to include short selling, borrowing, and other trading restrictions, see \cite{karatzasshreve:mathfinance}. 

There has been extensive research in dynamic portfolio optimization. The stochastic control approach was first introduced in the two landmark papers of Merton \cite{merton:lifetime, merton:consumption}, which was wedded to the Hamilton-Jacobi-Bellman equation and the requirement of an underlying Markov state process. The optimal consumption/investment problem in a non-Markov setting was solved using the martingale method by, among others, Pliska \cite{pliska:optimalportfolio}, Cox and Huang \cite{coxhuang:diffusion, coxhuang:variational}, Karatzas, Lehoczky and Shreve \cite{kls:smallinvestor}. The stochastic duality theory of Bismut \cite{bismut:convexdual} was first employed to study the constrained optimal investment problem in Shreve and Xu \cite{xushreve:dualmethod} where the authors studied the problems of no-short-selling constrains with $ K=[0,\infty)^{N} $. The effectiveness of convex duality method was later adopted to tackle the more traditional incomplete market models in the works of, among others, Karatzas, Lehoczky, Shreve and Xu\cite{klsx:dual}, Pearson and He \cite{hepearson:infinite, hepearson:finite}, Cvitani$\acute{c}$ and  Karatzas \cite{cvitanickaratzas:convexdual}. The spirit of this approach is to suitably embed the constrained problem in an appropriate family of unconstrained ones and find a member of this family for which the corresponding optimal policy obeys the constrains. However, despite the evident power of this approach, it is nevertheless true that obtaining the corresponding dual problem remains a challenge as it often involves clever experimentation and subsequently show to work as desired. To bring some transparency to the dual problem, Labb\'{e} and Heunis \cite{heunislabbe:constrainedUtility} established a simple synthetic method of arriving at a dual functional, bypassing the need to formulate a fictitious market. It often happens that the dual problem is much nicer than the primal problem in the sense that it is easier to show the existence of a solution and in some cases explicitly obtain a solution to the dual problem than it is to do likewise for the primal problem.

In this paper, we follow the approach as in Labb\'{e} and Heunis \cite{heunislabbe:constrainedUtility} by first converting the original problem into a static problem in an abstract space. Then we apply convex analysis to derive its dual problem and get the specific dual stochastic control problem. Subsequently, following the approach in \cite{santacroce:fbsdesemi} and \cite{horst:fbsde} we progress to a stochastic approach to simultaneously characterise the necessary and sufficient optimality conditions for both the primal and dual problems as systems of Forward and Backward Stochastic Differential Equations (FBSDEs) coupled with static optimality conditions. Such formulation then allows us to characterize the primal optimal control as a function of the adjoint processes coming from the dual FBSDEs in a dynamic fashion and vice versa. Moreover, we also find that the optimal primal wealth process coincides with the optimal adjoint process of the dual problem and vice versa. To the best of our knowledge, this is the first time the dynamic relations of the primal and dual problems have been explicitly established for constrained utility maximization problems under a non-Markov setting. After establishing the optimality conditions and the relations for the primal and dual problems, we solve three constrained utility maximization problems with both Markov and non-Markov setups. Instead of tackling the primal problem directly, we start from the dual problem and then construct the optimal solution to the primal problem from that to the dual problem. All examples contrast the  simplicity of the duality approach we propose and the technical complexity of solving the primal problem directly.

The rest of the paper is organised as follows. In Section 2 we set up the market model and formulate the primal and dual problems following the approach in \cite{heunislabbe:constrainedUtility}. In Section 3 we state and prove the main results of necessary and sufficient optimality conditions for the primal and dual problems and their connections in a dynamic fashion. In Section 4  we give three examples  to demonstrate the effectiveness of the dynamic duality approach in solving constrained utility maximization problems. Section 5 concludes the paper.

\section{Market Model and Primal and Dual Problems}
Let $ (\Omega,\mathcal{F},\mathbb{P}) $ be a complete probability space on which is defined some $\mathbb{R}^N$-valued standard Brownian motion $ \{W(t),t\in[0,T]\} $ with $ T>0 $ denoting a fixed terminal time. Let $ \{ \mathcal{F}_t, t\in[0,T] \} $ be the standard filtration induced by $ W $, where
\begin{equation*}
\mathcal{F}_t\triangleq \sigma\{ W(s),s\in[0,t] \}\bigvee \mathcal{N}(P), t\in[0,T],
\end{equation*}
in which $ \mathcal{N}(P) $ denotes the collection of all $ \mathbb{P} $-null events in $ (\Omega,\mathcal{F},\mathbb{P}) $. We denote by $ \mathcal{F}^* $ the $ \sigma $-algebra of $ \mathcal{F}_t $ progressively measurable sets on $ \Omega\times[0,T] $. For any stochastic process $ v: \Omega\times [0,T]\rightarrow\mathbb{R}^m, \ m\in\mathbb{Z}^+$, we write $ v\in\mathcal{F}^* $ to indicate $ v $ is $ \mathcal{F}^* $ measurable.  We introduce the following notations:
\begin{align*}
&\mathcal{H}^1(0,T;\mathbb{R}^m)\triangleq \left\lbrace v:\Omega\times[0,T] \rightarrow\mathbb{R}^m \mid v\in\mathcal{F}^*,E\left[\int_0^T \|v(t)\|dt\right]<\infty \right\rbrace,\\
&\mathcal{H}^2(0,T;\mathbb{R}^m)\triangleq \left\lbrace\xi:\Omega\times[0,T] \rightarrow\mathbb{R}^m \mid \xi\in\mathcal{F}^*,E\left[\int_0^T \|\xi(t)\|^2dt\right]<\infty \right\rbrace,
\end{align*}
where $ m\in\mathbb{Z}^+ $.\\

Consider a market consisting of a bank account with price $ \lbrace S_0(t) \rbrace $ given by
\begin{equation}\label{bank_account}
dS_0(t) = r(t)S_0(t)dt, \ 0\leq t\leq T, \ S_0(0)=1,
\end{equation}
and $ N $ stocks with prices $ \lbrace S_n(t) \rbrace, \ n=1,\cdots,N $, given by 
\begin{equation}\label{stock_price}
dS_n(t) = S_n(t)\left[ b_n(t)dt+\sum_{m=1}^N\sigma_{nm}(t)dW_{m}(t) \right],\ 0\leq t\leq T, \ S_n(0)>0.
\end{equation}
Through out the paper we assume that the interest rate $ \{r(t)\} $, the appreciation rates on stocks denoted by entries of the $ \mathbb{R}^N $-valued process $ \{b(t)\} $ and the volatility process denoted by entries of the $ N\times N $ matrix $ \{\sigma(t)\} $ are uniformly bounded $ \{\mathcal{F}_t\} $-progressively measurable scalar processes on $ \Omega\times [0,T] $. We also assume that there exists a positive constant $ k$ such that 
\begin{equation*}
z^\intercal\sigma(t)\sigma^\intercal(t)z\geq k|z|^2
\end{equation*}
for all $ (z,\omega,t)\in \mathbb{R}^N\times \Omega\times [0,T] $, where $z^\intercal$ is the transpose of $z$. According to \cite[p.90 (2.4) and (2.5)]{xushreve:dualmethod}, the strong non-degeneracy condition above ensures that matrices $\sigma(t), \sigma^\intercal(t)$ are  invertible and uniformly bounded.  \\

Consider a small investor with initial wealth $ x_0>0 $ and a self-financing strategy.  Define the set of admissible portfolio strategies by
\begin{equation*}
\left.
\begin{array}{c}
\mathcal{A}:= \left\{ \pi \in  {\cal H}^2(0,T; \mathbb{R}^N): \pi(t)\in K  \textit{ for $t\in[0,T]$ a.e.}\right\},
\end{array}
\right.
\end{equation*}
where $ K \subseteq \mathbb{R}^N $ is a closed convex set with $ 0\in K $ and  $ \pi$ is a  portfolio process with  each entry $ \pi_n(t) $ defined as the fraction of the investor's total wealth put into the stock $n$ for $n=1,\ldots,N$ at time $ t $.
Given any  $ \pi\in\mathcal{A} $, the investor's total wealth $ X^{\pi}$ satisfies the SDE
\begin{equation}\label{primal_state_process}
\left\{
\begin{array}{l}
dX^{\pi}(t)=X^{\pi}(t) \{[ r(t)+\pi^\intercal(t)\sigma(t)\theta(t) ] dt + \pi^\intercal(t)\sigma(t)dW(t)\},\quad 0\leq t\leq T,
\\
X^{\pi}(0)=x_0,
\end{array}
\right.
\end{equation}
where $\theta(t) := \sigma^{-1}(t)\left[ b(t)-r(t)\textbf{1} \right]$ is
the market price of risk at time $t$ and is uniformly bounded and  $ \textbf{1}\in\mathbb{R}^N $ has all unit entries.   A pair $(X,\pi)$ is  {\it admissible} if $\pi\in {\cal A}$ and $X$ is a strong solution to the SDE (\ref{primal_state_process}) with control process $\pi$.
\begin{remark}
Here we define the nth entry of $ \pi(t) $ as the fraction of small investor's wealth invested in the stock $ n $ at time $ t $. Such set-up ensures the positivity of the wealth process $ X^{\pi} $, but surrenders the Lipschitz property of the coefficients in both $ X $ and $ \pi $. Hence, the stochastic maximum principle developed in \cite{cadenillaskaratzas:convexsmp} and \cite{peng:smp} are not directly applicable in our case.
\end{remark}
Let $ U: [0,\infty)\rightarrow\mathbb{R}$ be a given utility function that is twice continuously differentiable, strictly increasing, strictly concave and satisfies the following conditions:
\begin{align*}
U(0)\triangleq\lim_{x\searrow 0} U(x)>-\infty, \, \lim_{x\searrow 0}U'(x)=\infty,\textit{ and }  \lim_{x\rightarrow \infty}U'(x)=0.
\end{align*}
Define the value of the expected utility maximization problem as 
\begin{equation*}
V\triangleq \sup_{\pi\in\mathcal{A}} E\left[ U\left( X^{\pi}(T) \right) \right].
\end{equation*}
To avoid trivialities, we assume that 
\begin{equation*}
-\infty<V<+\infty.
\end{equation*} 
The constrained utility maximization can be written as the following stochastic optimization problem:
\begin{equation*}
\textit{Find optimal }\pi^*\in\mathcal{A} \textit{ such that } E\left[ U\left( X^{\pi^*}(T)\right)\right]=\sup_{\pi\in\mathcal{A}} E\left[ U\left( X^{\pi}(T) \right) \right]=V.
\end{equation*}
In the rest of this section, we formulate the dual problem following the approach in \cite{heunislabbe:constrainedUtility}. Given any continuous $ \{\mathcal{F}_t\} $ semimartingale process $ X $, we write $ X\in \mathbb{R}\times\mathcal{H}^1(0,T;\mathbb{R})\times\mathcal{H}^2(0,T;\mathbb{R}^N) $ if
\begin{equation*}
X(t)=X_0+\int_0^t \dot{X}(s)ds+\int_0^t \Lambda_{X}^\intercal(s)dW(s), \; 0\leq t\leq T, 
\end{equation*}
where $(X_0,\dot X, \Lambda_{X})\in \mathbb{R}\times\mathcal{H}^1(0,T;\mathbb{R})\times\mathcal{H}^2(0,T;\mathbb{R}^N)$.
Define the following sets:
\begin{align*}
\mathcal{U}(X)&\triangleq \left\{ \pi\in\mathcal{A}| \dot{X}(t)=X(t)\left[ r(t)+\pi^\intercal(t)\sigma(t)\theta(t) \right] \textit{ and } \Lambda_x(t)=X(t)\sigma^\intercal(t)\pi(t) \ a.e.\right\},\\
\mathbb{B}&\triangleq \left\{ X\in \mathbb{R}\times\mathcal{H}^1(0,T;\mathbb{R})\times\mathcal{H}^2(0,T;\mathbb{R}^N) | X(0)=x_0\textit{ and }\mathcal{U}(x)=\varnothing  \right\}.
\end{align*}
Moreover, to remove the portfolio constraints, define the penalty functions:
\begin{eqnarray*}
  l_0(x)& \triangleq& \left \{
  \begin{aligned}
    &0, && \text{if}\ x=x_0,\\
    &+\infty, && \text{otherwise,}
  \end{aligned} \right. \\
  l_T(x) &\triangleq & \left \{
  \begin{aligned}
    &-U(x), && \text{if}\ x\in(0,\infty), \\
    &+\infty, && \text{otherwise,}
  \end{aligned} \right. \\
  L(t,x,v,\xi) &\triangleq& \left \{
  \begin{aligned}
    &0, && \text{if}\ x>0,v=xr(t)+\xi^\intercal\theta(t) \textit{ and }x^{-1}[\sigma^\intercal(t)]^{-1}\xi\in K,\\
    &+\infty, && \text{otherwise.}
  \end{aligned} \right.
\end{eqnarray*} 
Hence, following \cite[Remark 3.4]{heunislabbe:constrainedUtility}, we obtain
\begin{equation*}
-V=\inf_{X\in\mathbb{B}}\Phi_{p}(X)=\Phi_{p}(\hat{X}) \textit{ for some } \hat{X}\in\mathbb{B},
\end{equation*}
where 
$$ \Phi_{p}(X)\triangleq l_0(X(0))+E\left[ l_T(X(T)) \right] + E\int_0^T L(t,X(t),\dot{X}(t),\Lambda_X(t))dt. $$
 The dual problem is formulated in terms of the following pointwise convex conjugate transforms of the three penalty functions:
\begin{align*}
m_0(y)&\triangleq\sup_{x\in\mathbb{R}}[xy-l_0(x)]=x_0y,\\
m_T(y)&\triangleq\sup_{x\in\mathbb{R}}[x(-y)-l_T(y)]\\
&= \left \{
  \begin{aligned}
    &\tilde{U}(y)\triangleq\sup_{x>0}[U(x)-xy], && \text{if}\ y\in[0,\infty),\\
    &\infty, && \text{otherwise,}
  \end{aligned} \right.\\
M(t,y,s,\gamma)&\triangleq \sup_{x,v\in\mathbb{R},\ \xi\in\mathbb{R}^N}\{ xs+vy+\xi^\intercal\gamma-L(t,x,v,\xi) \}\\
&=\left \{
  \begin{aligned}
    &0, && \text{if}\ s+yr(t)+\delta_{K}(-\sigma(t)[y\theta(t)+\gamma]) <\infty,\\
    &\infty, && \text{otherwise,}
  \end{aligned} \right.
\end{align*}
where $ \delta_{K}(\cdot) $ is the support function of the set $ -K $ defined by
\begin{equation}\label{support_function}
\delta_{K}(z)\triangleq \sup_{\pi\in K}\{ -\pi^\intercal z \}, z\in\mathbb{R}^N.
\end{equation}
The dual objective function $ \Phi_D $ is given by
\begin{equation*}
\Phi_D(Y)\triangleq m_0(Y(0))+E\left[ m_T(Y(T)) \right]+E\int_0^T M(t,Y(t),\dot{Y}(t),\Lambda_Y(t))dt,
\end{equation*}
$  \forall Y\in \mathbb{R}\times\mathcal{H}^1(0,T;\mathbb{R})\times\mathcal{H}^2(0,T;\mathbb{R}^N). $
Define the set
\begin{equation*}
\mathcal{D}\triangleq\left\lbrace v\triangleq\Omega\times[0,T]\rightarrow\mathbb{R}^N |v\in\mathcal{F}^* \textit{ and } \int_0^T \left[ \delta_{K}(v(t))+\| v(t) \|^2 \right] dt <\infty \ a.s.\right\rbrace.
\end{equation*}
Given $ (y,v)\in (0,\infty)\times\mathcal{D} $, the corresponding state process $ Y^{(y,v)} $ satisfies the SDE
\begin{equation}\label{dual_state_process}
\left\{
\begin{array}{l}
dY^{(y,v)}(t)=-Y^{(y,v)}(t)\left\lbrace [r(t)+\delta_{K}(v(t))]dt+\left[ \theta(t)+\sigma^{-1}(t)v(t) \right]^\intercal dW(t) \right\rbrace,\ 0\leq t\leq T,
\\
Y^{(y,v)}(0)=y,
\end{array}
\right.
\end{equation}
The optimal value of the dual function is given by
\begin{equation*}
\tilde{V}\triangleq\inf_{(y,v)\in(0,\infty)\times\mathcal{D}}\left\lbrace x_0y+E\left[\tilde{U}(Y^{(y,v)}(T))\right] \right\rbrace.
\end{equation*}
The  dual problem can be written as the following stochastic optimization problem:
\begin{equation*}
\textit{Find the optimal }(y^*,v^*)\in(0,\infty)\times\mathcal{D} \textit{ such that } \tilde{V}=x_0y^*+E\left[\tilde{U}(Y^{(y^*,v^*)}(T))\right].
\end{equation*}
The duality relation follows from \cite[Corollary 4.12]{heunislabbe:constrainedUtility}. In this paper, instead of applying the convex duality method of \cite{bismut:convexdual}, we use the machinery of stochastic maximum principle and BSDEs to derive the necessary and sufficient conditions of the primal and dual problems separately. After establishing the optimal conditions as two systems of FBSDEs, we explicitly characterise the primal optimal solution as functions of the adjoint process coming from the dual FBSDEs in a dynamic fashion and vice versa.

\section{Main Results}
In this section, we derive the necessary and sufficient optimality conditions for the primal and dual problems and show that the connection between the optimal solutions through their corresponding FBSDEs.\\

Given an admissible control $ \pi\in\mathcal{A} $ and a solution $ X^{\pi} $ to the SDE \eqref{primal_state_process}, the associated adjoint equation in the unknown processes $ p_1\in\mathcal{H}^2 (0,T;\mathbb{R}) $ and $ q_1\in\mathcal{H}^2 (0,T;\mathbb{R}^N) $ is the following BSDE:
\begin{equation}\label{primal_adjoint_BSDE}
\left\{
\begin{array}{l}
dp_1(t)=-\left\lbrace \left[r(t)+\pi^\intercal(t)\sigma(t)\theta(t)\right]p_1(t)+q_1^\intercal(t)\sigma^\intercal(t)\pi(t) \right\rbrace dt+q_1^\intercal(t)dW(t),\vspace{2mm}\\
p_1(T)=-U'(X^{\pi}(T)).
\end{array}
\right.
\end{equation}
\begin{assumption}\label{primal_assump}
The utility function $ U $ satisfies the following conditions
\begin{itemize}
\item[(i)] $ x\rightarrow xU'(x) $ is non-decreasing on $ (0,\infty) $.
\item[(ii)] There exists $ \gamma\in(1,\infty) $ and $ \beta\in(0,1) $ such that $ \beta U'(x)\geq U'(\gamma x) $ for all $ x\in (0,\infty) $.
\end{itemize}
Moreover, we assume that for $ \forall \pi\in\mathcal{A} $ and corresponding $ X^{\pi} $ satisfying the SDE \eqref{primal_state_process}, $ E[| U(X^{\pi}(T)) |]<\infty $ and $ E\left[ \left( U'(X^{\pi}(T))X^{\pi}(T) \right)^2 \right]<\infty $.
\end{assumption}
\begin{remark}
The above assumption corresponds to Remark 3.4.4 in \cite{karatzasshreve:mathfinance}. Firstly, under Assumption \ref{primal_assump} (i), 
\begin{itemize}
\item[(i')] For a utility function $ U $ of class $ C^2(0,\infty) $ (which is true in our set-up), the Arrow-Pratt index of relative risk aversion $ R(x)=-\frac{xU''(x)}{U'(x)} $ does not exceed 1.
\end{itemize}
Moreover, set $ y=U'(x) $ and we have $ xU'(x)=yI(y)=-y\tilde{U}'(y) $. Hence, we conclude:
\begin{itemize}
\item[(ii')] $ z\rightarrow \tilde{U}(e^{z}) $ is convex in $ \mathbb{R} $ when $ \tilde{U} $ is the convex dual of $ U $.
\end{itemize}
Finally, replacing $ x $ by $ -\tilde{U}'(y) $, we claim that Assumption \ref{primal_assump} (ii) is equivalent to $ \tilde{U}'(\beta y)\geq \gamma \tilde{U}'(y) $ for $ \forall y\in(0,\infty) $ and some $ \beta\in (0,1), \gamma\in(1,\infty) $. Iterating the above inequality we obtain
\begin{itemize}
\item[(iii')] $\forall \beta\in(0,1) \ \exists \gamma\in(1,\infty) \  s.t \ \tilde{U}'(\beta y)\geq \gamma \tilde{U}'(y) $ for $ \forall y\in(0,\infty) $.\end{itemize}
\end{remark}
\begin{lemma}\label{primal_bsde_lemma}
Let $ \hat{\pi}\in\mathcal{A} $ and $ X^{\hat{\pi}} $  be a solution to the SDE \eqref{primal_state_process}. The there exists a solution $ (\hat{p}_1,\hat{q}_1) $ to the adjoint BSDE \eqref{primal_adjoint_BSDE}.
\end{lemma}
\begin{proof}
According to Assumption \ref{primal_assump}, the process defined as
\begin{equation}
\alpha(t)\triangleq E\left[ -X^{\hat{\pi}}(T)U'(X^{\hat{\pi}}(T) \big | \mathcal{F}_t \right],\ t\in [0,T]
\end{equation}
is square integrable. In addition, it is the unique solution of the BSDE
\begin{equation*}
\alpha(t)=-X^{\hat{\pi}}(T)U'(X^{\hat{\pi}}(T)) - \int_t^T \beta^\intercal(t)dW(t), \ t\in[0,T],
\end{equation*}
where $ \beta $ is a square integrable process with values in $ \mathbb{R}^N $. Applying Ito's lemma to $ \frac{\alpha(t)}{X^{\hat{\pi}}(t)} $, we have
\begin{eqnarray*}
d\dfrac{\alpha(t)}{X^{\hat{\pi}}(t)} &= &\dfrac{\beta(t)}{X^{\hat{\pi}}(t)}dW(t)-\dfrac{\alpha(t)}{X^{\hat{\pi}}(t)}\left\{ [r(t)+\hat{\pi}^\intercal(t)\sigma(t)\theta(t)]dt+\hat{\pi}^\intercal(t)\sigma(t)dW(t) -| \pi^\intercal(t)\sigma(t) |^2 dt\right\} \\
&&{}- \dfrac{\hat{\pi}^\intercal(t)\sigma(t)\beta(t)}{X^{\hat{\pi}}(t)}dt\\
&=&-\left\{ [ r(t)+\hat{\pi}^\intercal(t)\sigma(t)\theta(t)]\hat{p}_1(t) +\hat{q}_1^\intercal(t)\sigma^\intercal(t)\hat{\pi}(t)\right\}dt + \hat{q}_1^\intercal(t)dW(t),
\end{eqnarray*}
where 
\begin{equation}\label{primal_adjoint_process}
\hat{p}_1(t)\triangleq \dfrac{\alpha(t)}{X^{\hat{\pi}}(t)} \textit{ and } \hat{q}_1(t)\triangleq \dfrac{\beta(t)}{X^{\hat{\pi}}(t)}-\dfrac{\alpha(t)\sigma^\intercal(t)\hat{\pi}(t)}{X^{\hat{\pi}}(t)}.
\end{equation}
Hence, we conclude that $ (\hat{p}_1,\hat{q}_1) $ solves the adjoint BSDE \eqref{primal_adjoint_BSDE}.
\end{proof}

We now state the necessary and sufficient conditions for the optimality of primal problem.
\begin{theorem}(Primal problem and associated FBSDE)\label{primal_thm}
Let $ \hat{\pi}\in\mathcal{A} $. Then  $ \hat{\pi}$ is  optimal for the primal problem  if and only if
the solution  $ (X^{\hat{\pi}},\hat{p}_1,\hat{q}_1) $  of FBSDE 
\begin{equation}\label{primal_fbsde}
\left\{
\begin{array}{l}
dX^{\hat{\pi}}(t)=X^{\hat{\pi}}(t) \{[ r(t)+\hat{\pi}^\intercal(t)\sigma(t)\theta(t) ] dt + \hat{\pi}^\intercal(t)\sigma(t)dW(t)\},\vspace{2mm}
\\
X^{\hat{\pi}}(0)=x_0,\vspace{2mm}
\\
dp_1(t)=-\left\lbrace \left[r(t)+\hat{\pi}^\intercal(t)\sigma(t)\theta(t)\right]p_1(t)+q_1^\intercal(t)\sigma^\intercal(t)\hat{\pi}(t) \right\rbrace dt+q_1^\intercal(t)dW(t),\vspace{2mm}\\
p_1(T)=-U'(X^{\hat{\pi}}(T))
\end{array}
\right.
\end{equation}
satisfies the condition 
\begin{equation}\label{primal_cond}
-X^{\hat{\pi}}(t)\sigma(t)\left[ \hat{p}_1(t)\theta(t)+\hat{q}_1(t) \right]\in N_K(\hat{\pi}(t)) \textit{ for }\forall t\in[0,T], \ \mathbb{P}-a.s.,
\end{equation}
where $ N_K(x) $ is the normal cone to the closed convex set $ K $ at $ x \in K$, defined as
\begin{equation*}
N_K(x)\triangleq\left\lbrace y\in\mathbb{R}^N:\forall x^*\in K, y^\intercal(x^*-x)\leq 0 \right\rbrace.
\end{equation*}
\end{theorem}
\begin{proof}
Let $ \tilde{\pi}\in\mathcal{A} $ be an admissible control and $ \rho\triangleq \tilde{\pi}-\hat{\pi} $. Let $ \tau_n\triangleq T \wedge \inf\Big\{t\geq 0, \int_0^t \| \rho(s)\sigma(s) \|^2 ds \geq n \textit{ or } \int_0^t \| \rho^\intercal(s)\sigma(s)\sigma^\intercal(s)\hat{\pi}(s) \|^2 ds \geq n \Big\} $. Hence, $ \lim_{n\rightarrow \infty}\tau_n=T $ almost surely. Define $ \rho_n(t)\triangleq \rho(t)1_{\{t \leq \tau_n\}} $. Define function $ \phi_n (\varepsilon)\triangleq U\left( X^{\hat{\pi}+\varepsilon\rho_n}(T) \right) $ where $ \varepsilon\in [0,1] $. Set $ G(x)\triangleq U(x_0e^x) $ and taking derivatives, we have
\begin{align*}
&G'(x)=U'(x_0e^x)x_0e^x\geq 0,\\
&G''(x)=x_0e^x\left( U'(x_0e^x)+U'(x_0e^x)x_0e^x \right) \leq 0,
\end{align*}
by Assumption \ref{primal_assump}. Differentiating $ \phi $ on $ (0,1) $, we have
\begin{align*}
\phi'_n(\varepsilon)=&G'(\cdot)\left[ \int_0^T \left( \rho_n^\intercal(t)\sigma(t)\theta(t)-\rho_n^\intercal(t)\sigma^\intercal(t)\sigma(t)\left( \hat{\pi}(t)+\varepsilon\rho_n(t) \right) \right) dt+\int_0^T\rho_n^\intercal(t)\sigma(t)dW(t)\right].\\
\phi''_n(\varepsilon)=&G''(\cdot)\left[ \int_0^T \int_0^T \left( \rho_n^\intercal(t)\sigma(t)\theta(t)-\rho_n^\intercal(t)\sigma^\intercal(t)\sigma(t)\left( \hat{\pi}(t)+\varepsilon\rho_n(t) \right) \right) dt+\int_0^T\rho_n^\intercal(t)\sigma(t)dW(t) \right]^2\\
&-G'(\cdot)\left[\int_0^T \rho_n^\intercal(t)\sigma^\intercal(t)\sigma(t)\rho_n(t) dt\right]\leq 0.
\end{align*}
Hence we conclude that the function $ \Phi_n(\varepsilon)\triangleq\frac{\phi_n(\varepsilon)-\phi(0)}{\varepsilon} $ is a decreasing function and we have
\begin{equation}
\lim_{\varepsilon\rightarrow 0}\Phi_n(\varepsilon)=U'(X^{\hat{\pi}}(T))X^{\hat{\pi}}(T)H_n^\rho(T),
\end{equation}
where $ H_n^\rho(t)\triangleq \int_0^t \left(\rho_n^\intercal(s)\sigma(s)\theta(s)-\rho_n^\intercal(s)\sigma^\intercal(s)\sigma(s)\hat{\pi}(s) \right) ds+\int_0^t\rho_n^\intercal(s)\sigma(s)dW(s)$. Moreover, we obtain
\begin{align*}
E\left[ | U'(X^{\hat{\pi}}(T))X^{\hat{\pi}}(T)H^{\rho}_n(T) | \right]\leq E\left[ \left(U'(X^{\hat{\pi}}(T))X^{\hat{\pi}}(T)\right)^2 \right]^{\frac{1}{2}}E\left[ H^{\rho}_n(T)^2 \right]^{\frac{1}{2}}<\infty.
\end{align*}
Note that for $ \forall \varepsilon \in[0,1], \Phi_n(\varepsilon)\geq \Phi_n(1) = U(X^{\hat{\pi}+\rho_n}(T))-U(X^{\hat{\pi}}(T)) $ with $ E\left[\Phi_n(1)\right]<\infty $. Therefore the sequence $ \Phi_n(\varepsilon) $ is bounded from below. By the Monotone Convergence Theorem, we have
\begin{equation*}
\lim_{\varepsilon\rightarrow 0}\dfrac{E\left[ U(X^{\hat{\pi}+\varepsilon\rho_n}(T)) \right]-E\left[ U(X^{\hat{\pi}}(T)) \right]}{\varepsilon}=E\left[ U'(X^{\hat{\pi}}(T))X^{\hat{\pi}}(T)H^{\rho}_n(T)\right].
\end{equation*}
In addition, since $ \hat{\pi} $ is optimal, we conclude 
\begin{equation}\label{primal_nece_ineq_1}
E\left[ U'(X^{\hat{\pi}}(T))X^{\hat{\pi}}(T)H^{\rho}_n(T)\right]\leq 0.
\end{equation}
Let $ (\alpha,\beta) $ be as defined in Lemma \ref{primal_bsde_lemma} and $ (\hat{p}_2),\hat{q}_2) $ be the adjoint process corresponding to $ \hat{\pi} $. Applying Ito's lemma to $ -\alpha(t) H^\rho_n(t) $, we have
\begin{align*}
-d\alpha(t) H_n^\rho(t) = &\beta^\intercal(t)H^\rho(t)dW(t)-\alpha(t)\left( \rho_n^\intercal(t)\sigma(t)\theta(t)-\rho_n^\intercal(t)\sigma(t)\sigma^\intercal(t)\hat{\pi}(t) \right)dt\\
&-\alpha(t)\rho_n^\intercal(t)\sigma(t)dW(t)+\rho_n^\intercal(t)\sigma(t)\beta(t)dt\\
=&\left[ -\hat{p}_1(t)X^{\hat{\pi}}(t)\rho_n^\intercal(t)\sigma(t)\Big(\theta(t)-\sigma^\intercal(t)\hat{\pi}(t)\Big) \right.\\
&\left.-\rho_n^\intercal(t)\sigma(t)\left( X^{\hat{\pi}}(t)\hat{q}_1(t)+X^{\hat{\pi}}(t)\hat{p}_1(t)\sigma^\intercal(t)\hat{\pi}(t) \right)\right]dt\\
&+\left[\beta^\intercal(t)H_n^\rho(t)-\alpha(t)\rho_n^\intercal(t)\sigma(t)\right]dW(t).
\end{align*}
Rearranging the above equation, we obtain
\begin{equation}
-d\alpha(t) H_n^\rho(t)=-X^{\hat{\pi}}(t)\rho_n^\intercal(t)\sigma(t)\bigg( \hat{p}_1(t)\theta(t)+\hat{q}_1(t) \bigg)dt +\left[\beta^\intercal(t)H^\rho_n(t)-\alpha(t)\rho_n^\intercal(t)\sigma(t)\right]dW(t).
\end{equation}
Next, we prove that the local martingale $ \int_0^t (\beta^\intercal(s)H^\rho(s)-\alpha(s)\rho^\intercal(s)\sigma(s) )dW(s) $ is indeed a true martingale.
\begin{align*}
&E\left[\sup_{t\in[0,T]}| H_n^\rho(t) |^2\right]\\
&=E\left[\sup_{t\in[0,T]}\bigg|\int_0^t \left(\rho_n^\intercal(s)\sigma(s)\theta(s)-\rho_n^\intercal(s)\sigma(s)\sigma^\intercal(s)\hat{\pi}(s) \right) ds+\int_0^t\rho_n^\intercal(s)\sigma(s)dW(s)\bigg|^2\right]\\
&\leq C\left\{E\left[\sup_{t\in[0,T]}\bigg|\int_0^t\rho_n^\intercal(s)\sigma(s)dW(s)\bigg|^2\right]
+E\left[\sup_{t\in[0,T]}\bigg|\int_0^t \left(\rho_n^\intercal(s)\sigma(s)\theta(s)-\rho_n^\intercal(s)\sigma(s)\sigma^\intercal(s)\hat{\pi}(s) \right) ds\bigg|^2\right]\right\}\\
&\leq C \left\{ E\left[ \int_0^T | \rho_n^\intercal(s)\sigma(s) |^2ds \right]+E\left[ \int_0^T|\rho_n^\intercal(s)\sigma(s)|^2ds\right]+E\left[ \int_0^T| \rho_n^\intercal(s)\sigma(s)\sigma^\intercal(s)\hat{\pi}(s) |^2ds \right] \right\}\\
&< \infty,
\end{align*}
by the Burkeholder-Davis-Gundy inequality. In addition, we have
\begin{equation*}
E\left[ \int_0^T | \alpha(s)\rho_n^\intercal(s)\sigma(s) |^2 ds\right]<\infty.
\end{equation*}
Hence, \eqref{primal_nece_ineq_1} can be reduced to the following
\begin{equation}\label{primal_nece_proof_cond}
E\left[ \int_0^{\tau_n} -X^{\hat{\pi}}(t)\rho_n^\intercal(t)\sigma(t)\bigg( \hat{p}_1(t)\theta(t)+\hat{q}_1(t) \bigg)dt \right]\leq 0 \ \forall n\in\mathbb{N}.
\end{equation}
Define the following sets:
\begin{equation*}
B\triangleq \left\{ (t,\omega)\in[0,T]\times \Omega: \left( \pi^\intercal-\hat{\pi}^\intercal(t) \right)\sigma(t)\left( \hat{p}_1(t)\theta(t)+\hat{q}_1(t) \right)<0, \textit{ for }\forall \pi\in K \right\},
\end{equation*}
and, for any $ \pi\in K $,
\begin{equation*}
B^{\pi}\triangleq \left\{ (t,\omega)\in[0,T]\times \Omega: \left( \pi^\intercal-\hat{\pi}^\intercal(t) \right)\sigma(t)\left( \hat{p}_1(t)\theta(t)+\hat{q}_1(t) \right)<0 \right\}.
\end{equation*}
Obviously for each $ t\in[0,T] $, $ B_t^{\pi}\in\mathcal{F}_t $. Consider the control $ \tilde{\pi}:[0,T]\times\Omega \rightarrow K $, defined by
\[
\tilde{\pi}(t,\omega)\triangleq
\begin{cases}
\pi, & \textit{ if } (t,\omega)\in B^{\pi}\\
\hat{\pi}(t,\omega), &\textit{ otherwise.}
\end{cases}
\]
Then $ \tilde{\pi} $ is adapted and $ \exists n^*\in\mathbb{N} $ such that
\begin{equation*}
E\left[ \int_0^{\tau_n} X^{\tilde{\pi}}(t)\left( \tilde{\pi}^\intercal(t)-\hat{\pi}^\intercal(t) \right)\sigma(t)\left( \hat{p}_1(t)\theta(t)+\hat{q}_1(t) \right)dt \right]<0 \ \forall n>n^*,
\end{equation*}
contradicting \eqref{primal_nece_proof_cond}, unless $(Leb\otimes\mathbb{P})\{B^{\pi}\}=0$ for $ \forall \pi\in K $. Since $ \mathbb{R}^N $ is a separable metric space, we can find a  countable dense subset $ \{\pi_n\} $ of  $ K $. Denote by $ \hat B=\cup_{n=1}^\infty B^{\pi_n} $. Then  $(Leb\otimes\mathbb{P})\{\hat B\}=(Leb\otimes\mathbb{P})\{\cup_{n=1}^\infty B^{\pi_n}\} \leq \sum_{n=1}^\infty (Leb\otimes\mathbb{P})\{B^{\pi_n}\}=0$. Hence, we conclude that
\begin{equation*}
-X^{\hat{\pi}}(t)\sigma(t)\left[ \hat{p}_1(t)\theta(t)+\hat{q}_1(t) \right]\in N_K(\hat{\pi}(t)) \textit{ for }\forall t\in[0,T], \ \mathbb{P}-a.s.
\end{equation*}
We have proved the necessary condition. 

Now we prove the sufficient condition. 
Let $ (X^{\hat{\pi}},\hat{p}_1,\hat{q}_1) $ be a solution to the FBSDE \eqref{primal_fbsde} and satisfy condition \eqref{primal_cond}. 
Applying Ito's formula, we have
\begin{align*}
&\left( X^{\hat{\pi}}(t)-X^{\pi}(t) \right)\hat{p}_1(t)\\
=&\int_0^t \left( X^{\hat{\pi}}(s)-X^{\pi}(s) \right)\left\lbrace  -\left[ \left( r(s)+\hat{\pi}^\intercal(s)\sigma(s)\theta(s) \right)\hat{p}_1(s)+\hat{q}_1^\intercal(t)\sigma^\intercal(t)\pi(t) \right]dt+\hat{q}_1^\intercal(t)dW(t)\right\rbrace\\
&+\int_0^t \hat{p}_1^\intercal(s)\bigg\{\left[ X^{\hat{\pi}}(t)\left( r(s)+\hat{\pi}^\intercal(s)\sigma(s)\theta(s) \right)-X^{\pi}(s)\left( r(s)+\pi^\intercal(s)\sigma(s)\theta(s) \right) \right]ds\\
&+\left[ X^{\hat{\pi}}(s)\hat{\pi}^\intercal(s)\sigma(s)-X^{\hat{\pi}}(s)\pi^\intercal(s)\sigma(s) \right]dW(s)\bigg\}\\
&+\int_0^t \left[ X^{\hat{\pi}}(s)\hat{\pi}^\intercal(s)\sigma(s)-X^{\pi}(s)\pi^\intercal(s)\sigma(s) \right]\hat{q}_1(s)ds.
\end{align*}
Rearranging the above equation, we have
\begin{align*}
&\left( X^{\hat{\pi}}(t)-X^{\pi}(t) \right)\hat{p}_1(t) \\
=&\int_0^t\left(\hat{\pi}^\intercal(s)-\pi^\intercal(s)\right) X^{\hat{\pi}}(s)\sigma(s)\left[ \hat{p}_1(s)\theta(s)+\hat{q}_1(s) \right] ds\\
&+\int_0^t \left[\left( X^{\hat{\pi}}(s)-X^\pi(s) \right) q^\intercal(s)+X^{\hat{\pi}}(s)\left( \hat{\pi}^\intercal(t)-\pi^\intercal(t) \right)\sigma(s) \right]dW(s).
\end{align*}
Hence, by Condition \eqref{primal_cond} and the definition of normal cone, taking expectation of the above, we have
\begin{equation*}
E\left[ \left( X^{\hat{\pi}}(T)-X^{\pi}(T) \right)\hat{p}_1(T) \right]\leq 0.
\end{equation*}
Combining with concavity of $ U $ gives us 
\begin{align*}
&E\left[ U\left(X^{\pi}(T)\right)-U\left(X^{\hat{\pi}}(T)\right) \right] \leq  E\left[ \left( X^{\pi}(T)-X^{\hat{\pi}}(T) \right)U'\left(X^{\hat{\pi}}(T)\right)\right] \\
& =E\left[ \left( X^{\hat{\pi}}(T)-X^{\pi}(T) \right)\hat{p}_1(T) \right]\leq 0.
\end{align*}
Hence $ \hat{\pi} $ is indeed an optimal control.
\end{proof}

Next we address the dual problem. To establish the existence of an optimal solution, we impose the following condition:
\begin{assumption}\label{existence_condition} (\cite[Condition 4.14]{heunislabbe:constrainedUtility})
For any $ (y,v)\in (0,\infty)\times\mathcal{D} $, we have $ E\left[ \tilde{U}\left( Y^{(y,v)}(T) \right)^2 \right]<\infty $.
\end{assumption}
According to \cite[Proposition 4.15]{heunislabbe:constrainedUtility}, there exists some $ (\hat{y},\hat{v})\in(0,\infty)\times \mathcal{D}$ such that $ \tilde{V}=x_0\hat{y}+E\left[ \tilde{U}\left( Y^{(\hat{y},\hat{v})}(T) \right) \right] $. Given admissible control $ (\hat{y},\hat{v})\in(0,\infty)\times\mathcal{D} $ with the state process $ Y^{(y,v)} $ that solves the SDE \eqref{dual_state_process} and $ E\left[ \tilde{U}\left( Y^{(y,v)}(T) \right)^2 \right]<\infty $, the associated adjoint equation for the dual problem is the following linear BSDE in the unknown processes $ p_2\in \mathcal{H}^2(0,T;\mathbb{R}) $ and $ q_2\in\mathcal{H}^2(0,T;\mathbb{R}^N) $:
\begin{equation}\label{dual_adjoint_BSDE}
\left\{
\begin{array}{l}
dp_2(t)=\left\lbrace \left[r(t)+\delta_K(v(t))\right]^\intercal p_2(t)+q_1^\intercal(t)\left[ \theta(t)+\sigma^{-1}(t)v(t) \right] \right\rbrace dt+q_2^\intercal(t)dW(t),\vspace{2mm}\\
p_2(T)=-\tilde{U}'(Y^{(y,v)}(T)).
\end{array}
\right.
\end{equation}
\begin{lemma}\label{dual_BSDE_existence_lemma}
Let $ (y,v) \in(0,\infty)\times\mathcal{D}$ and $ Y^{(y,v)} $ be the corresponding state process satisfying the SDE \eqref{dual_state_process} with $ E\left[ \tilde{U}\left( Y^{(y,v)}(T) \right)^2 \right]<\infty $. Then the random variable $ Y^{(y,v)}(T)\tilde{U}'(Y^{(y,v)}(T)) $ is square integrable and there exists a solution to the adjont BSDE \eqref{dual_adjoint_BSDE}.
\end{lemma}
\begin{proof}
According to Assumption \ref{existence_condition}, we have $ E\left[ \tilde{U}\left( Y^{(\hat{y},\hat{v})}(T) \right)^2 \right]<\infty. $  Following similar arguments as in \cite[page 290]{karatzasshreve:mathfinance} we have that since $ \tilde{U} $ is a decreasing function
\begin{align*}
\tilde{U}(\eta)-\tilde{U}(\infty)&\geq \tilde{U}(\eta)-\tilde{U}(\dfrac{\eta}{\beta})\\
&=\int_\eta^{\frac{\eta}{\beta}} I(u)du\\
&\geq \left( \dfrac{\eta}{\beta}-\eta \right)I\left(\dfrac{\eta}{\beta} \right)\\
&\geq \dfrac{1-\beta}{\beta\gamma}\eta I(\eta),
\end{align*}
for $ 0<\eta<\infty $, where $ \beta\in(0,1) $ and $ \gamma\in (1,\infty) $ are as in Condition \ref{existence_condition}. Since $\tilde U(\infty)=U(0)$ is finite, we conclude that the random variable $ Y^{(\hat{y},\hat{v})}(T)\tilde{U}'(Y^{(\hat{y},\hat{v})}(T)) $ is square integrable.  Define the process
\begin{equation*}
\phi(t)\triangleq E\left[ -Y^{(\hat{y},\hat{v})}(T)\tilde{U}'(Y^{(\hat{y},\hat{v})}(T)) \bigg | \mathcal{F}_t \right], \ t\in [0,T].
\end{equation*}
By the martingale representation theorem, it is the unique solution to the BSDE
\begin{equation*}
\phi(t)=-Y^{(\hat{y},\hat{v})}(T)\tilde{U}'(Y^{(\hat{y},\hat{v})}(T))-\int_t^T \varphi^\intercal(s)dW(s),
\end{equation*}
where $ \varphi $ is a square integrable process with values in $ \mathbb{R}^N $. Applying Ito's formula to $\frac{\phi(t)}{Y^{(\hat{y},\hat{v})}(t)}$, we have
\begin{align*}
d\dfrac{\phi(t)}{Y^{(\hat{y},\hat{v})}(t)} =& \left\{\frac{\phi(t)}{Y^{(\hat{y},\hat{v})}(t)}\left[ r(t)+\delta_K(\hat{v}(t))+| \theta(t)+\sigma^{-1}(t)\hat{v}(t) |^2 \right]+\frac{\varphi(t)}{Y^{(\hat{y},\hat{v})}(t)}[\theta(t)+\sigma(t)^{-1}\hat{v}(t)]\right\}dt\\
&+\left\{ \frac{\phi(t)}{Y^{(\hat{y},\hat{v})}(t)}\left[ \theta(t)+\sigma(t)^{-1}\hat{v}(t) \right]^\intercal+\frac{\varphi^\intercal(t)}{Y^{(\hat{y},\hat{v})}(t)} \right\}dW(t).
\end{align*}
Rearranging the above equation, we have
\begin{equation*}
d\hat{p}_2(t)=\left\{ \left[ r(t)+\delta_K(\hat{v}(t)) \right]^\intercal\hat{p}_2(t)+\hat{q}_2(t)\left[ \theta(t)+\sigma(t)^{-1}\hat{v}(t) \right]^\intercal \right\}dt+\hat{q}_2^\intercal(t)dW(t),
\end{equation*}
where $ (\hat{p}_2,\hat{q}_2) $ are defined as
\begin{equation*}
\hat{p}_2(t)\triangleq \dfrac{\phi(t)}{Y^{(\hat{y},\hat{v})}(t)} \textit{ and } \hat{q}_2(t)\triangleq\hat{p}_2(t)\left[ \theta(t)+\sigma(t)^{-1}\hat{v}(t) \right]^\intercal+\dfrac{\varphi^\intercal(t)}{Y^{(\hat{y},\hat{v})}(t)}.
\end{equation*}
Hence, we conclude that $ (\hat{p}_2,\hat{q}_2) $ solves the BSDE \eqref{dual_adjoint_BSDE}.
\end{proof}

\begin{remark} Note that if $U(x)=\ln x$ then $\tilde U(y)=-\ln y -1$. We have  $ Y^{(\hat{y},\hat{v})}(T)\tilde{U}'(Y^{(\hat{y},\hat{v})}(T))\equiv -1$, obviously square integrable. The conclusion of Lemma \ref{dual_BSDE_existence_lemma} holds. However, in this case, $U(0) = -\infty$, not finite. So the condition $U(0)$ being finite is only a sufficient condition for Lemma \ref{dual_BSDE_existence_lemma}, but not a necessary condition. We can apply all the results in the paper to log utility.
\end{remark}

We now state the necessary and sufficient conditions of optimality of the dual problem.
\begin{theorem}\label{dual_thm}(Dual problem and associated FBSDE)
Let $ (\hat{y},\hat{v})\in(0,\infty)\times\mathcal{D} $. Then $ (\hat{y},\hat{v})$ is optimal for the dual problem  if and only if  the solution  $ (Y^{(\hat{y},\hat{v})},\hat{p}_2,\hat{q}_2) $ of FBSDE
\begin{equation}\label{dual_FBSDE}
\left\{
\begin{array}{l}
dY^{(\hat{y},\hat{v})}(t)=-Y^{(\hat{y},\hat{v})}(t)\left\lbrace [r(t)+\delta_K(\hat{v}(t))]dt+[\theta(t)+\sigma^{-1}(t)\hat{v}(t)]^\intercal dW(t)\right\rbrace,\vspace{2mm}\\
Y^{(\hat{y},\hat{v})}(0)=\hat{y},\vspace{2mm}\\
dp_2(t)=\left\lbrace \left[r(t)+\delta_K(v(t))\right]^\intercal p_2(t)+q_2^\intercal(t)\left[ \theta(t)+\sigma^{-1}(t)v(t) \right] \right\rbrace dt+q_2^\intercal(t)dW(t),\vspace{2mm}\\
p_2(T)=-\tilde{U}'(Y^{(y,v)}(T))
\end{array}
\right.
\end{equation}
satisfies the following conditions
\begin{equation}\label{dual_conditions}
\left\{
\begin{array}{l}
\hat{p}_2(0)=x_0,\\
\hat{p}_2(t)^{-1}\left[\sigma^\intercal(t)\right]^{-1}\hat{q}_2(t)\in K,\\
\hat{p}_2(t)\delta_K(\hat{v}(t))+\hat{q}'_2(t)\sigma^{-1}(t)\hat{v}(t)=0, \textit{ for }\forall t\in [0,T]\ \mathbb{P}-a.s.
\end{array}
\right.
\end{equation}
\end{theorem}

\begin{proof}
Let $ (\hat{y},\hat{v}) $ be an optimal control of the dual problem and $ Y^{(\hat{y},\hat{v})} $ be the corresponding state process. Define function
$h(\xi)\triangleq x_0\xi\hat{y}+E\left[ \tilde{U}\left(\xi Y^{(\hat{y},\hat{v})}(T)\right) \right]$ , and $ \inf_{\xi\in (0,\infty)}h(\xi) = h(1) $. Then following the argument in \cite[Lemma 11.7, page 725]{klsx:dual} by the convexity of $ \tilde{U} $, the dominated convergence theorem and Lemma \ref{dual_BSDE_existence_lemma}, we conclude that $ h(\cdot) $ is continuously differentiable at $ \xi=1 $ and the derivative $ h'(1)=x_0\hat{y}+E\left[  Y^{(\hat{y},\hat{v})}(T)\tilde{U}'\left(Y^{(\hat{y},\hat{v})} (T)\right) \right] $ holds. Hence, we conclude that
\begin{equation}
\hat{p}_2(0)=-\frac{1}{\hat{y}}E\left[  Y^{(\hat{y},\hat{v})}(T)\tilde{U}'\left(Y^{(\hat{y},\hat{v})} (T)\right) \right]=x_0.
\end{equation}
Let $ (\hat{y},\tilde{v}) $ be an admissible control and $ \eta\triangleq \tilde{v}-\hat{v}$. Similar to the argument in \cite[page 781-782]{cvitanickaratzas:convexdual}, let the stopping time $ \tau_n\triangleq T \wedge \inf \{ t\in[0,T]; \int_0^t \| \delta_K(\eta(s)) \|^2+\| \theta^\intercal(s)\sigma^{-1}(s)\eta(s) \|^2+\| \hat{v}^\intercal(s)[\sigma^{-1}(s)]^\intercal\sigma^{-1}(s)\eta(s) \|^2 +\| \phi(s)\eta(s) \|^2 ds\geq n \textit{ or } \Big| \int_0^t \eta^\intercal(s)[\sigma^{-1}(s)]^\intercal dW(s)\Big|\geq n\} $. Let $ \eta_{n}(t)\triangleq \eta(t)1_{t\leq \tau_{n}} $. Define function $ \tilde{\phi}_n(\varepsilon)\triangleq \tilde{U}\left( Y^{(\hat{y},\hat{v}+\varepsilon\eta_n)}(T) \right) = \tilde{U}\left\{ \exp\left[\ln\left( Y^{(\hat{y},\hat{v}+\varepsilon\eta_n)}(T)\right)\right] \right\}$. According to Assumption \ref{primal_assump}, $ g(z)\triangleq\tilde{U}(e^z) $ is a convex function that is non-increasing. Moreover, since $ \delta_K $ is convex, $ f(\varepsilon)\triangleq\ln\left( Y^{(\hat{y},\hat{v}+\varepsilon\eta_n)}(T)\right) $ is a concave function of $ \varepsilon $. Hence $ \tilde{\phi}_n(\varepsilon) = g(f(\varepsilon)) $ is a convex function and  $ \tilde{\Phi}_n(\varepsilon)\triangleq\frac{\tilde{\phi}_n(\varepsilon)-\tilde{\phi}_n(0)}{\varepsilon} $ is an increasing function. Define $ \tilde{H}^{\eta_n}_\varepsilon(t) $ and $ \tilde{H}^{\eta_n}(t) $ as
\begin{align*}
\tilde{H}^{\eta_n}_\varepsilon({t}) \triangleq  &\int_0^{t}\delta_K(\hat{v}(s)+\varepsilon\eta_n(s))-\delta_K(\hat{v}(s))+\varepsilon\theta^\intercal(s)\sigma^{-1}(s)\eta_n(s)+\varepsilon\hat{v}^\intercal(s)[\sigma^{-1}(s)]^\intercal\sigma^{-1}(s)\eta_n(s)\\
&+\frac{1}{2}\varepsilon^2\eta_n^\intercal(s)[\sigma^{-1}(s)]^\intercal\sigma^{-1}(s)\eta_n(s)ds+\int_0^{t}\varepsilon\eta_n^\intercal(s)[\sigma^{-1}(s)]^\intercal dW(s),\\
\tilde{H}^{\eta_n}({t}) \triangleq & \int_0^{t}\delta_K(\eta_n(s))+\theta^\intercal(s)\sigma^{-1}(s)\eta_n(s)+\hat{v}^\intercal(s)[\sigma^{-1}(s)]^\intercal\sigma^{-1}(s)\eta_n(s)ds\\
&+\int_0^{t}\eta_n^\intercal(s)[\sigma^{-1}(s)]^\intercal dW(s).
\end{align*}

Let $ \varepsilon\in(0,1) $, we have
\begin{align*}
\tilde{\Phi}_n(\varepsilon)=&\frac{\tilde{U}\left( Y^{(\hat{y},\hat{v}+\varepsilon\eta_n)}({T}) \right)-\tilde{U}\left( Y^{(\hat{y},\hat{v})}({T}) \right) }{\varepsilon}\\
=&\frac{\tilde{U}\left( Y^{(\hat{y},\hat{v}+\varepsilon\eta_n)}({T}) \right)-\tilde{U}\left( Y^{(\hat{y},\hat{v})}({T}) \right) }{Y^{(\hat{y},\hat{v}+\varepsilon\eta_n)}({T})-Y^{(\hat{y},\hat{v})}({T})}\dfrac{Y^{(\hat{y},\hat{v})}({T})}{\varepsilon}\left[ \dfrac{Y^{(\hat{y},\hat{v}+\varepsilon\eta_n)}({T})}{Y^{(\hat{y},\hat{v})}({T})} -1 \right]\\
=&\frac{\tilde{U}\left( Y^{(\hat{y},\hat{v}+\varepsilon\eta_n)}({T}) \right)-\tilde{U}\left( Y^{(\hat{y},\hat{v})}({T}) \right) }{Y^{(\hat{y},\hat{v}+\varepsilon\eta_n)}({T})-Y^{(\hat{y},\hat{v})}({T})}\dfrac{Y^{(\hat{y},\hat{v})}({T})}{\varepsilon}\left[ \exp\left( -\tilde{H}^{\eta_n}_\varepsilon({T})  \right) -1 \right]\\
\leq &\frac{\tilde{U}\left( Y^{(\hat{y},\hat{v}+\varepsilon\eta_n)}({T}) \right)-\tilde{U}\left( Y^{(\hat{y},\hat{v})}({T}) \right) }{Y^{(\hat{y},\hat{v}+\varepsilon\eta_n)}({T})-Y^{(\hat{y},\hat{v})}({T})}\dfrac{Y^{(\hat{y},\hat{v})}({T})}{\varepsilon}\\
&\left\{ -1+\exp\left[ -\varepsilon \int_0^{T}\bigg(\delta_K(\eta_n(t))+\theta^\intercal(t)\sigma^{-1}(t)\eta_n(t)+\hat{v}^\intercal(t)[\sigma^{-1}(t)]^\intercal\sigma^{-1}(t)\eta_n(t) \right.\right.\\
& \left.\left.+\frac{1}{2}\varepsilon\eta_n^\intercal(t)[\sigma^{-1}(t)]^\intercal\sigma^{-1}(t)\eta_n(t)\bigg)dt -\varepsilon \int_0^{T}\eta_n^\intercal(t)[\sigma^{-1}(t)]^\intercal dW(t) \right] \right\}.
\end{align*}
Hence, taking $ \limsup $ on both sides, we have
\begin{equation*}
\limsup_{\varepsilon\rightarrow 0}\tilde{\Phi}_n(\varepsilon)\leq -\tilde{U}'\left(Y^{(\hat{y},\hat{v})}({T})\right)Y^{(\hat{y},\hat{v})}({T})\tilde{H}^{\eta_n}({T})
\end{equation*}
with
\begin{equation*}
E\left[ \bigg|\tilde{U}'\left(Y^{(\hat{y},\hat{v})}({T})\right)Y^{(\hat{y},\hat{v})}({T})\tilde{H}^{\eta_n}({T}) \bigg|\right]\leq E\left[\left(\tilde{U}'\left(Y^{(\hat{y},\hat{v})}({T})\right)Y^{(\hat{y},\hat{v})}({T})\right)^2\right]^{\frac{1}{2}}E\left[\tilde{H}^{\eta_n}({T})^2\right]^{\frac{1}{2}}< \infty.
\end{equation*}
Moreover, notice that as $ \varepsilon\in(0,1) $ approaches zero, the sequence $ \left(\frac{\tilde{U}\left( Y^{(\hat{y},\hat{v}+\varepsilon\eta_n)}({T}) \right)-\tilde{U}\left( Y^{(\hat{y},\hat{v})}({T}) \right) }{\varepsilon}\right)_{\varepsilon\in (0,1)} $ is bounded from above by $| \tilde{\Phi}_n(1) |$ and $E\left[| \tilde{\Phi}_n(1) |\right]<\infty$.  By the reverse Fatou lemma, we have
\begin{equation*}
0\leq \limsup_{\varepsilon\rightarrow 0}E\left[\tilde{\Phi}_n(\varepsilon)\right] \leq E\left[\limsup_{\varepsilon\rightarrow 0}\tilde{\Phi}_n(\varepsilon) \right] \leq E\left[ -\tilde{U}'\left(Y^{(\hat{y},\hat{v})}({T})\right)Y^{(\hat{y},\hat{v})}({T})\tilde{H}^{\eta_n}({T}) \right].
\end{equation*}


Let $ (\phi,\varphi) $ be defined as in Lemma \ref{dual_BSDE_existence_lemma} and $ (\hat{p}_2,\hat{q}_2) $ be the adjoint process corresponding to $ (\hat{y},\hat{v}) $. Apply Ito's lemma to $ \phi(t)\tilde{H}_n^\eta(t) $, we obtain
\begin{align*}
&d\phi(t)\tilde{H}^{\eta_n}(t)\\
=&-\varphi^\intercal(t)\tilde{H}^{\eta_n}(t)dW(t)+\phi(t)\left( \delta_K(\eta_n(t))+\theta^\intercal(t)\sigma^{-1}(t)\eta_n(t)+\hat{v}(t)[\sigma^{-1}(t)]^\intercal\sigma^{-1}(t)\eta_n(t) \right)dt\\
&+\phi(t)\eta_n^\intercal(t)[\sigma^{-1}(t)]^\intercal dW(t)+\eta_n^\intercal(t)[\sigma^{-1}(t)]^\intercal\varphi(t)dt\\
=&Y^{\hat{y},\hat{v}}(t)\left[ \delta_K(\eta_n(t))\hat{p}_2(t)+\hat{q}_2(t)\sigma^{-1}(t)\eta_n(t)) \right]dt+\left[ \phi(t)\eta_n^\intercal(t)[\sigma^{-1}(t)]^\intercal-\varphi^\intercal(t)\tilde{H}^{\eta_n}(t) \right]dW(t)
\end{align*}
Following similar approach as in the proof of necessary condition for the primal problem, it can be shown that $ \int_0^{t} \left[ \phi(s)\eta_n^\intercal(s)[\sigma^{-1}(s)]^\intercal-\varphi^\intercal(s)\tilde{H}^{\eta_n}(s) \right]dW(s) $ is a true martingale. Taking expectation of the above equation, we obtain
\begin{equation}\label{dual_proof_contra}
E\left[\int_0^{\tau_n} Y^{(\hat{y},\hat{v})}(t)\left[ \delta_K(\eta(t))\hat{p}_2(t)+\hat{q}_2(t)\sigma^{-1}(t)\eta(t)) \right]dt\right]\geq 0.
\end{equation}
Note that $ \hat{p}_2(t)=\frac{\phi(t)}{Y^{(\hat{y},\hat{v})}(t)}>0 $, define the event $ B\triangleq\left\{ (\omega,t): \hat{p}_2(t)^{-1}\sigma(t)^{-1}\hat{q}_2(t)\not\in K \right\} $. According to \cite[Lemma 5.4.2 on page 207]{karatzasshreve:mathfinance}, there exists some $ \mathbb{R}^N $ valued progressively measurable process $ \eta $ such that $ \|\eta(t)\|\leq 1 $ and $\| \delta_{K}(\eta(t)) \|\leq 1$ a.e. and 
\begin{align*}
\delta_K(\eta(t))+\hat{p}_2(t)^{-1}\hat{q}_2(t)'\sigma(t)^{-1}<0 \textit{ a.e. on } B,\\
\delta_K(\eta(t))+\hat{p}_2(t)^{-1}\hat{q}_2(t)'\sigma(t)^{-1}=0 \textit{ a.e. on } B^c.
\end{align*}
Let $ \tilde{v}\triangleq \hat{v}+\eta $. We can easily verify that $ \tilde{v} $ is progressively measurable and square integrable. Hence, we obtain that
\begin{equation*}
E\left\{ \int_0^{\tau_n} Y^{(\hat{y},\hat{v})}(t)\left[ \hat{p}_2(t)\left( \delta_K(\eta(t))\right)+\hat{q}_2(t)'\sigma(t)^{-1}\eta(t) \right]dt \right\} < 0,
\end{equation*}
contradicting with \eqref{dual_proof_contra}.
Hence, by the $ \mathbb{P} $ strict positivity of $ Y^{(\tilde{y},\tilde{v})}(t)\hat{p}_2(t) $, we conclude that $ \hat{p}_2(t)^{-1}\sigma(t)^{-1}\hat{q}_2(t)\in K $ a.e. (this argument is essentially identical to the analysis in the proof of Proposition 4.17 in \cite{heunislabbe:constrainedUtility}). Take $ \tilde{v}=2\hat{v} $, and we have
\begin{equation}
E\left\{ \int_0^{\tau_n} Y^{(\tilde{y},\tilde{v})}(t)\left[ \hat{p}_2(t)\left( \delta_K(\hat{v}(t))\right)+\hat{q}_2(t)'\sigma(t)^{-1}\hat{v}(t) \right]dt \right\} \geq 0.
\end{equation}
Lastly, to prove the third condition, simply take $ \tilde{v}=0 $ and by the same analysis, we obtain
\begin{equation*}
E\left\{ \int_0^{\tau_n} Y^{(\tilde{y},\tilde{v})}(t)\left[ \hat{p}_2(t)\left( \delta_K(\hat{v}(t))\right)+\hat{q}_2(t)'\sigma(t)^{-1}\hat{v}(t) \right]dt \right\} \leq 0.
\end{equation*}
On the other hand, by the definition of $ \delta_K $, we have $ \delta_K(\hat{v}(t))+\hat{p}_2(t)^{-1}\hat{q}_2^\intercal(t)\sigma^{-1}(t)\hat{v}(t)\geq 0 $ a.e. Combining with the $ \mathbb{P} $ strict positivity of $ Y^{(\tilde{y},\tilde{v})}(t)\hat{p}_2(t) $ gives the last condition.
We have proved the necessary condition. 

Now we prove the sufficient condition. 
Let $ \left( Y^{(\hat{y},\hat{v})},\hat{p}_2,\hat{q}_2 \right) $ be a solution to the FBSDE \eqref{dual_FBSDE} and satisfy conditions \eqref{dual_conditions}. Let the pair $ (\tilde{y},\tilde{v})\in(0,\infty)\times\mathcal{D} $ be a given admissible control such that $ Y^{(\hat{y},\hat{v})} $ solves the SDE \eqref{dual_state_process} and $ E\left[ \tilde{U}(Y^{(\tilde{y},\tilde{v})}(T))^2 \right]<\infty $. By Lemma \ref{dual_BSDE_existence_lemma}, we claim that there exists adjoint process $ (\tilde{p}_2,\tilde{q}_2) $ that solves the BSDE with control $ (\tilde{y},\tilde{v}) $.
Applying Ito's formula, we have
\begin{align*}
&\left( Y^{(\hat{y},\hat{v})}(t)-Y^{(\tilde{y},\tilde{v})}(t)\right)\hat{p}_2(t)\\
=&\hat{p}_2(0)y+\int_0^t\left\lbrace Y^{(\tilde{y},\tilde{v})}(s)\left[r(s)+\delta_K(\tilde{v}(s))\right]^\intercal-Y^{(\hat{y},\hat{v})}(s)\left[r(s)+\delta_K(\hat{v}(s))\right]^\intercal\right\rbrace\hat{p}_2(s)ds\\
&+\int_0^t\left\lbrace Y^{(\tilde{y},\tilde{v})}(s)[\theta(s)+\sigma^{-1}(s)\tilde{v}(s)]^\intercal-Y^{(\hat{y},\hat{v})}(s)[\theta(s)+\sigma^{-1}(s)\hat{v}(s)]^\intercal\right\rbrace\hat{p}_2(s)dW(s)\\
&+\int_0^t\left( Y^{(\hat{y},\hat{v})}(s)-Y^{(\tilde{y},\tilde{v})}(s) \right)\left\lbrace \left[r(s)+\delta_K(\tilde{v}(s))\right]^\intercal\hat{p}_2(s)+\hat{q}^\intercal_2(s)\left[ \theta(s)+\sigma^{-1}(s)\hat{v}(s) \right]\right\rbrace ds\\
&+\int_0^t \left( Y^{(\hat{y},\hat{v})}(s)-Y^{(\tilde{y},\tilde{v})}(s) \right)\hat{q}^\intercal_2(s) dW(s)\\
&+\int_0^t \left\lbrace Y^{(\tilde{y},\tilde{v})}(s)[\theta(s)+\sigma^{-1}(s)\tilde{v}(s)]^\intercal-Y^{(\hat{y},\hat{v})}(s)[\theta(s)+\sigma^{-1}(s)\hat{v}(s)]^\intercal\right\rbrace\hat{q}_2(s)ds\\
=&\hat{p}_2(0)y+\int_0^t Y^{(\tilde{y},\tilde{v})}(s)\hat{p}_2(s)\left[ \delta_K(\tilde{v}(s))-\delta_K(\hat{v}(s))+\hat{q}^\intercal_2(s)\sigma^{-1}(s)\left( \tilde{v}(s)-\hat{v}(s) \right) \right]ds\\
&+\int_0^t \left\lbrace Y^{(\tilde{y},\tilde{v})}(s)[\theta(s)+\sigma^{-1}(s)\tilde{v}(s)]^\intercal-Y^{(\hat{y},\hat{v})}(s)[\theta(s)+\sigma^{-1}(s)\hat{v}(s)]^\intercal \right\rbrace \hat{p}_2(s)dW(s)\\
&+\int_0^t \left( Y^{(\hat{y},\hat{v})}(s)-Y^{(\tilde{y},\tilde{v})}(s) \right)\hat{q}^\intercal_2(s) dW(s) .
\end{align*}
By \eqref{dual_conditions} and taking expectation, we have
\begin{align*}
E\left[\left(Y^{(\hat{y},\hat{v})}(T)-Y^{(\tilde{y},\tilde{v})}(T)\right)\hat{p}_2(T)\right]\geq & y\hat{p}_2(0).
\end{align*}
By convexity of $ \tilde{U} $ we obtain
\begin{equation*}
x_0\tilde{y}+E\left[ \tilde{U}(Y^{(\tilde{y},\tilde{v})}(T)) \right]-x_0\hat{y}-E\left[\tilde{U}(Y^{(\hat{y},\hat{v})}(T))\right]\geq y(x_0-\hat{p}_2(0)) = 0.
\end{equation*}
Hence, we conclude that $ (\hat{y},\hat{v}) $ is indeed an optimal control of the dual problem.
\end{proof}
We can now state the dynamic relations of the primal portfolio and wealth processes of the primal problem and the adjoint processes of the dual problem and vice versa.
\begin{theorem}\label{thm_dual_primal}(From dual problem to primal problem)
Suppose that $ (\hat{y},\hat{v}) $ is optimal for the dual problem. Let $ \left( Y^{(\hat{y},\hat{v})},\hat{p}_2,\hat{q}_2 \right) $ be the associated process that solves the FBSDE \eqref{dual_FBSDE} and satisfies condition \eqref{dual_conditions}. Define
\begin{equation}\label{primal_control_by_dual}
\hat{\pi}(t)\triangleq \dfrac{\left[\sigma^\intercal(t)\right]^{-1}\hat{q}_2(t)}{\hat{p}_2(t)},\ t\in[0,T].
\end{equation}
Then $ \hat{\pi} $ is the optimal control for the primal problem with initial wealth $ x_0 $. The optimal wealth process and associated adjoint process are given by
\begin{equation}\label{primal_proess_by_dual}
\left\{
\begin{array}{l}
X^{\hat{\pi}}(t)=\hat{p}_2(t),\vspace{1mm}\\
\hat{p}_1(t)=-Y^{(\hat{y},\hat{v})}(t),\vspace{1mm}\\
\hat{q}_1(t)=Y^{(\hat{y},\hat{v})}(t)[\sigma^{-1}(t)\hat{v}(t)+\theta(t)].
\end{array}
\right.
\end{equation}
\end{theorem}
\begin{proof}
Suppose that $ (\hat{y},\hat{v})\in(0.\infty)\times\mathcal{D} $ is optimal for the dual problem. By Theorem \ref{dual_thm}, the process $ \left(Y^{(\hat{y},\hat{v})},\hat{p}_2,\hat{q}_2\right) $ solves the dual FBSDE \eqref{dual_FBSDE} and satisfies condition \eqref{dual_conditions}. Construct $ \hat{\pi} $ and $ (X^{\hat{\pi}},\hat{p}_1,\hat{q}_1) $ as in \eqref{primal_control_by_dual} and \eqref{primal_proess_by_dual}, respectively. Substituting back into the \eqref{primal_fbsde}, we conclude that $ (X^{\hat{\pi}},\hat{p}_1,\hat{q}_1) $ solves the FBSDE for the primal problem. Moreover, by \eqref{dual_conditions} it can be easily shown that $ \hat{\pi}\in\mathcal{A} $ and \eqref{primal_cond} holds. By condition \eqref{dual_conditions}, it can be easily shown that $ \pi\in\mathcal{A} $. Moreover, we have
\begin{align*}
&X^{\hat{\pi}}(t)\sigma(t)\left[ \hat{p}_1(t)\theta(t)+\hat{q}_1(t) \right]\\
&=\hat{p}_2(t)\sigma(t)\left\{ -Y^{(\hat{y},\hat{v})}(t)\theta(t)+Y^{(\hat{y},\hat{v})}(t)\left[ \sigma^{-1}(t)\hat{v}(t)+\theta(t) \right] \right\}\\
&=Y^{(\hat{y},\hat{v})}(t)\hat{p}_2(t)\hat{v}(t).
\end{align*}
Combing with the third statement of \eqref{dual_conditions} and the almost surely positivity of $ Y^{(\hat{y},\hat{v})}\hat{p}_2 $, we claim that condition \eqref{primal_cond} holds.  By Theorem \ref{dual_thm} we conclude that $ \hat{\pi} $ is indeed an optimal control to the primal problem.
\end{proof}
\begin{theorem}(From primal problem to dual problem)
Suppose that $ \hat{\pi}\in\mathcal{A} $ is optimal for the primal problem with initial wealth $ x_0 $. Let $ (X^{\hat{\pi}},\hat{p}_1,\hat{q}_1) $ be the associated process that satisfies the FBSDE \eqref{primal_fbsde} and conditions \eqref{primal_cond}. Define
\begin{equation}\label{dual_control_by_primal}
\left\{
\begin{array}{l}
\hat{y}\triangleq-\hat{p}_1(0),\\
\hat{v}(t)\triangleq-\sigma(t)\left[ \dfrac{\hat{q}_1(t)}{\hat{p}_1(t)}+\theta(t) \right], \textit{ for }\forall t\in[0,T].
\end{array}
\right.
\end{equation}
Then $ (\hat{y},\hat{v}) $ is an optimal control for the dual problem. The optimal dual state process and associated adjoint process are given by
\begin{equation}\label{dual_process_by_primal}
\left\{
\begin{array}{l}
Y^{(\hat{y},\hat{v})}(t)=-\hat{p}_1(t),\vspace{1mm}\\
\hat{p}_2(t)=X^{\hat{\pi}}(t),\vspace{1mm}\\
\hat{q}_2(t)=\sigma^\intercal(t)\hat{\pi}(t)X^{\hat{\pi}}(t).
\end{array}
\right.
\end{equation}
\end{theorem}
\begin{proof}
Suppose that $ \hat{\pi}\in\mathcal{A} $ is an optimal control for the primal problem. By Theorem \ref{primal_thm}, the process $ (X^{\hat{\pi}},\hat{p}_1,\hat{q}_1) $ solves that FBSDE \eqref{primal_fbsde} and satisfies conditions \eqref{primal_cond}. Define $ (\hat{y},\hat{v}) $ and $ (Y^{(\hat{y},\hat{v})},\hat{p}_2,\hat{q}_2) $ as in \eqref{dual_control_by_primal} and \eqref{dual_process_by_primal}, respectively. Substituting them back into \eqref{dual_FBSDE}, we obtain that $ (Y^{(\hat{y},\hat{v})},\hat{p}_2,\hat{q}_2) $ solves the FBSDE for the dual problem. Moreover, by the construction in \eqref{dual_control_by_primal} and \eqref{dual_process_by_primal}, we have $ \hat{p}_2(0)=x_0 $ and $ [\sigma^\intercal(t)]^{-1}\hat{q}_2(t)=\hat{\pi}(t)X^{\hat{\pi}}(t)^{-1}\in K $. Substituting $\hat{v}$ into \eqref{dual_conditions}, we can easily show that the third statement in \eqref{dual_conditions} holds. Hence, by Theorem \ref{dual_thm}, we conclude that $ (\hat{y},\hat{v}) $ is indeed an optimal control to the dual problem.
\end{proof}
\section{Examples}
In this section, we shall use the results introduced in previous sections to address several classical constrained utility maximization problems.
\subsection{Constrained Power Utility Maximization}
In this subsection, we assume $ U $ is a power utility function defined by $ U(x)\triangleq\frac{1}{\beta}x^{\beta}, \ x\in(0,\infty) $, where $ \beta \in(0,1) $ is a constant. In addition, we assume that $ K\subseteq \mathbb{R}^N $ is a closed convex cone. In this case, the dual problem can be written as
\begin{equation*}
\textit{Minimize }x_0y+E\left[ \tilde{U}\left( Y^{(y,v)}(T) \right) \right]
\end{equation*}
over  $ (y,v)\in(0,\infty)\times\mathcal{D} $, where $ \tilde{U}(y)=\frac{1-\beta}{\beta}y^{\frac{\beta}{\beta-1}},\ y\in(0,\infty) $. We solve the above problem in two steps: first fix $ y $ and find the optimal control $ \hat{v}(y) $; second find the optimal $ \hat{y} $. We can then construct the optimal solution explicitly.
\begin{description}
\item [Step 1:] Consider the associated HJB equation:
\begin{equation}\label{power_utility_HJB}
\left\{
\begin{array}{l}
v_t(s,y)-r(s)yv_y(s,y)+\frac{1}{2}\inf_{v\in \tilde{K}}\| \sigma^{-1}(s)v+\theta(s) \|^2y^2v_{yy}(s,y)=0 \vspace{1mm}\\
v(T,y)=\frac{1-\beta}{\beta}y^{\frac{\beta}{\beta-1}},
\end{array}
\right.
\end{equation}
for each $ (s,y)\in[t,T]\times\mathbb{R} $. The infimum term in \eqref{power_utility_HJB} can be written explicitly as $ \hat{v}(s)=\sigma(s)\textit{proj}[-\theta(s)|\sigma^{-1}(s)\tilde{K}] $.
Then the HJB equation \eqref{power_utility_HJB} becomes
\begin{equation*}
\left\{
\begin{array}{l}
v_t(s,y)-r(s)yv_y(s,y)+\frac{1}{2}y^2\theta_{v}^2(s)v_{yy}(s,y)=0 \vspace{1mm}\\
v(T,y)=\frac{1-\beta}{\beta}y^{\frac{\beta}{\beta-1}},
\end{array}
\right.
\end{equation*}
where $ \theta_{\hat{v}}(s)=\theta(s)+\sigma^{-1}(s)\hat{v}(s) $.\\
According to the Feynman-Kac formula, we have
\begin{align*}
v(t,y)=E\left[ \frac{1-\beta}{\beta}Y^{\frac{\beta}{\beta-1}}(T) \right]=\frac{1-\beta}{\beta}y^{\frac{\beta}{\beta-1}}\exp{\left\{\int_t^T\left[\frac{1}{2}\frac{\beta}{(\beta-1)^2}\theta_v^2(s)-\frac{\beta}{\beta-1}r(s)\right]ds\right\}},
\end{align*}
where the stochastic process $ Y $ follows the geometric Brownian motion
\begin{equation*}
dY(t)=-Y(t)[r(t)dt+\theta_{v}(t)dW(t)],\ Y(0)=y.
\end{equation*}
In particular, we have $ v(0,y)=y^{\frac{\beta}{\beta-1}}\exp{\left\{\int_0^T\left[\frac{1}{2}\frac{\beta}{(\beta-1)^2}\theta_{\hat{v}}^2(s)-\frac{\beta}{\beta-1}r(s)\right]ds\right\}} $

\item [Step 2:] Solving the following static optimization problem 
\begin{equation*}
\inf_{y\in\mathbb{R}} x_0y+y^{\frac{\beta}{\beta-1}}\exp{\left\{\int_0^T\left[\frac{1}{2}\frac{\beta}{(\beta-1)^2}\theta_{\hat{v}}^2(s)-\frac{\beta}{\beta-1}r(s)\right]ds\right\}},
\end{equation*}
we obtain
\begin{equation}
\hat{y}=x_0^{\beta-1}\exp{\left\{(1-\beta)\int_0^T\left[ \frac{\beta}{2(\beta-1)^2}\theta_{v}^2(s)-\frac{\beta}{\beta-1}r(s) \right]ds\right\}}.
\end{equation}
\end{description}
Solving the adjoint BSDE, we have
\begin{align}\label{power_utility_p1}
\hat{p}_2(t)&=x_0\exp{\int_0^t\left[r(s)+\frac{(1-2\beta)}{2(1-\beta)^2}\theta_{\hat{v}}(s)^2\right]ds+\frac{1}{1-\beta}\int_0^t\theta_{\hat{v}}(s)dW(s)},\\\label{power_utility_q1}
\hat{q}_2(t)&=\dfrac{\theta_{\hat{v}}(t)}{1-\beta}\hat{p}_2(t).
\end{align}
Applying Theorem \ref{thm_dual_primal}, we can construct the optimal solution to the primal problem using the optimal solutions of the dual problem and hence arrive at the following closed form solutions:
\begin{equation*}
\left\{
\begin{array}{l}
\hat{\pi}(t)=[\sigma(t)^\intercal]^{-1}\dfrac{\theta_{\hat{v}}(t)}{1-\beta},\\
X^{\hat{\pi}}(t)=x_0\exp{\left\{\displaystyle\int_0^t\left[r(s)+\frac{(1-2\beta)}{2(1-\beta)^2}\theta_{\hat{v}}(s)^2\right]ds+\frac{1}{1-\beta}\int_0^t\theta_{\hat{v}}(s)dW(s)\right\}}.
\end{array}
\right.
\end{equation*}
\subsection{Constrained Log Utility Maximization with Random Coefficients}
In this section, we assume that $ U $ is a log utility function defined by $ U(x)=\log x$ for $ x>0 $. The dual function of $ U $ is defined as $ \tilde{U}(y)\triangleq -(1+\log y), \ y\geq 0 $. Assume that $ K\subseteq \mathbb{R}^N $ is a closed convex set and $ r,b,\sigma $ are uniformly bounded $ \{\mathcal{F}_t\} $ progressively measurable processes on $ \Omega\times [0,T] $.
\begin{description}
\item[Step 1:] We fix $ y $ and attempt to solve for the optimal control $ \hat{v}(y) $. Note that the dynamic programming technique is not appropriate in this case due to the non-Markov nature of the problem. However, following the approach in \cite[Section 11, p.790]{cvitanickaratzas:convexdual} the problem can be solved explicitly due to the special property of the logarithmic function.

Let $ v\in\mathcal{D} $ be any given admissible control and the objective function becomes
\begin{equation*}
x_0y + E\left[\tilde{U}\left( Y^{(y,v)}(T) \right)\right]=x_0y-1-\log y - E\left[ \int_0^T r(t)+\delta_{K}(v(t))+\frac{1}{2}\| \theta(t)+\sigma(t)v(t) \|^2 dt \right].
\end{equation*}
The dual optimization boils down to the following problem of pointwise minimization of a convex function $\delta_{K}(v)+\frac{1}{2}\| \theta(t)+\sigma(t)v \|^2$ over $ v\in\tilde{K} $ for $ \forall t\in[0,T] $. Applying classical measurable selection theorem (see \cite{schal:aselectionthem} and \cite{schal:condoptm}), we conclude that the process defined by
\begin{equation}\label{egpointwisemin_1}
\hat{v}(t)\triangleq \argmin_{v\in\tilde{K}}\left[ \delta_{K}(v)+\frac{1}{2}\| \theta(t)+\sigma(t)^{-1}v \|^2 \right]
\end{equation}
is $ \{\mathcal{F}_t\} $ progressively measurable and therefore is the optimal control given $ y $.

\item[Step 2:] Solve the following static optimization problem
\begin{equation*}
\inf_{y\in\mathbb{R}} x_0y-1-\log y - E\left[ \int_0^T r(t)+\delta_{K}(\hat{v}(t))+\frac{1}{2}\| \theta(t)+\sigma(t)v(t) \|^2 dt \right].
\end{equation*}
We obtain $ \hat{y}=\frac{1}{x_0} $. Hence, the optimal state process for the dual problem is the exponential process satisfying \eqref{dual_state_process}.
\end{description}
Solving the adjoint BSDE \eqref{dual_adjoint_BSDE}, we have
\begin{equation}
\hat{p}_2(t)Y^{(\hat{y},\hat{v})}(t)=E\left[ -Y^{(\hat{y},\hat{v})}(T)\tilde{U}\left(Y^{(\hat{y},\hat{v})}(T)\right) \bigg |\mathcal{F}_t\right]=1.
\end{equation}
Hence, we have $  \hat{p}_2(t) = Y^{(\hat{y},\hat{v})}(T)^{-1}$. Applying Ito's formula on $ \hat{p}_2 $, we have
\begin{equation*}
\hat{q}_2(t)=Y^{(\hat{y},\hat{v})}(t)^{-1}[\theta(t)+\sigma(t)^{-1}\hat{v}(t)]\textit{ for }\forall t\in[0,T], \textit{ a.e}.
\end{equation*}
Finally, according to Theorem \ref{thm_dual_primal}, we construct the optimal control to the primal problem explicitly form the optimal solution of the dual problem as
\begin{equation}\label{eglogprimalcontrol}
\hat{\pi}(t)=[\sigma(t)\sigma^\intercal(t)]^{-1}\left[ \hat{v}(t)+b(t)-r(t)\mathbf{1} \right]\textit{ for }\forall t\in[0,T], \textit{ a.e}.
\end{equation}
\begin{remark}
In the case where $ K $ is a closed convex cone, it is trivial to see that $ \delta_{K}(\hat{v}(t))=0 $ for $ \forall t\in[0,T] $. Then the pointwise minimization problem \eqref{egpointwisemin_1} becomes a simple constrained quadratic minimization problem
\begin{equation*}
\hat{v}(t)\triangleq \argmin_{v\in\tilde{K}}\| \theta(t)+\sigma(t)^{-1}v \|^2, \forall t\in[0,T].
\end{equation*}
Furthermore, in the case where $ K=\mathbb{R}^N $ and $ \hat{v}=0 $, the optimal control \eqref{eglogprimalcontrol} reduces to $ \hat{\pi}(t)=[\sigma(t)\sigma^\intercal(t)]^{-1}\left[b(t)-r(t)\mathbf{1} \right] \textit{ for }\forall t\in[0,T]$, and we recover the unconstrained log utility maximization problem discussed in \cite{karatzas:optmtrading}.
\end{remark}
\begin{remark}
From the above two examples, we contrast our method to the approach in \cite{cvitanickaratzas:convexdual,klsx:dual,karatzasshreve:mathfinance}, which rely on the introduction of a family of auxiliary unconstrained problems formulated in auxiliary markets parametrized by money market and stock mean return rates \cite[see Section 8]{cvitanickaratzas:convexdual}. The existence of a solution to the original problem is then equivalent to finding the fictitious market that provides the correct optimal solution to the primal problem. On the other hand, we explicitly write our the dual problem to the original constrained problem only relying on elementary convex analysis results and characterize its solution in terms of FBSDEs. The dynamic relationship between the primal and dual FBSDEs then allows us to explicitly construct optimal solution to the primal problem from that to the dual problem.
\end{remark}
\subsection{Constrained Non-HARA Utility Maximization}
In this subsection, we assume $ U $ is a Non HARA utility function defined by $ U(x)=\frac{1}{3}H(x)^{-3}+H(x)^{-1}+xH(x) $ for $x>0$, where $ H(x)=\left(\frac{2}{-1+\sqrt{1+4x}}\right)^{\frac{1}{2}} $. The dual function of $ U $ is defined as $ \tilde{U}\triangleq\sup_{x>0}[U(x)-xy]=\frac{1}{3}y^{-3}+y^{-1},\ y\in [0,\infty).$ Assume that $ K\subseteq \mathbb{R}^N $ is a closed convex cone and $ r,b,\sigma $ are constants. Hence, the dual problem becomes
\begin{equation*}
\textit{Minimize }x_0y+E\left[ \frac{1}{3}\left( Y^{(y,v)}(T) \right)^{-3}+\left( Y^{(y,v)}(T) \right)^{-1} \right] \textit{ over }(y,x)\in (0,\infty)\times \mathcal{D}.
\end{equation*}
We solve the above problem in two steps: first, fix $ y $ and find the optimal control $ \tilde{v}(y) $; second, find the optimal $ \hat{y} $. We can then construct the optimal solution explicitly.
\begin{description}
\item [Step 1:] Consider the associated HJB equation:
\begin{equation}\label{non_hara_utility_HJB}
\left\{
\begin{array}{l}
v_t(s,y)-ryv_y(s,y)+\frac{1}{2}\inf_{v\in \tilde{K}}\| \sigma^{-1}v+\theta \|^2y^2v_{yy}(s,y)=0,\textit{ where } (s,y)\in (0,T)\times[0,\infty),\vspace{2mm}\\
v(T,y)=\frac{1}{3}y^{-3}+y^{-1},
\end{array}
\right.
\end{equation}
for each $ (s,y)\in[t,T]\times[0,\infty) $. Let $ \hat{v} $ be the minimizer of $ \inf_{v\in\tilde{K}}| \theta+\sigma^{-1}v |^2 $ and $ \hat{\theta}\triangleq\theta+\sigma^{-1}\hat{v} $. Define $ w(\tau,y)\triangleq v(s,y) $ with $ \tau=T-s $. We have $ w $ solves the following PDE:
\begin{equation}\label{non_hara_utility_HJB}
\left\{
\begin{array}{l}
w_t(\tau,y)+ryw_y(\tau,y)-\frac{1}{2}\hat{\theta}^2y^2w_{yy}(\tau,y)=0,\textit{ where } (\tau,y)\in (0,T)\times(0,\infty),\vspace{2mm}\\
w(0,y)=\frac{1}{3}y^{-3}+y^{-1},
\end{array}
\right.
\end{equation}
Next, we follow the approach in \cite{bianzheng:JEDC}, we solve the  above PDE. Let $ \alpha=\frac{1}{2}+\frac{r}{\hat{\theta}^2}, \ a=\frac{1}{\sqrt{2}}\hat{\theta}, \beta=-a^2\alpha^2 $, and $\hat{w}(s,z)=e^{-az+\beta s}w(t,e^{z}) $ , then $ \hat{w} $ solves the heat equation $ \hat{w}_t-a^2\hat{w}_{zz}=0 $ and has  the initial condition $ \hat{w}(0,z)=e^{-az}\left( \frac{e^{-3z}}{3}+e^{-z} \right) $. Using Poisson's formula to find $ w(s,z) $ and $ v(s,y) $, we have
\begin{equation*}
v(s,y)=\frac{1}{3}y^{-3}e^{3r(T-s)+6\hat{\theta}^2(T-s)}+\frac{1}{y}e^{r(T-s)+\hat{\theta}^2(T-s)}.
\end{equation*}
\item[Step 2:] Considering the following static optimization problem:
\begin{equation}\label{non_hara_static}
\inf_{y\in(0,\infty)} x_0y+\frac{1}{3}y^{-3}e^{3rT+6\hat{\theta}^2T}+\frac{1}{y}e^{rT+\hat{\theta}^2T}.
\end{equation}
Solving \eqref{non_hara_static}, we have
\begin{equation}
-\hat{y}^{-4}e^{3rT+6\hat{\theta}^2T}-\hat{y}^{-2}e^{rT+\hat{\theta}^2T}+x_0=0.
\end{equation}
Hence, we have $ \hat{y} = \frac{1}{\sqrt{2x_0}}\left[ e^{(r+\hat{\theta}^2)T}+\sqrt{e^{2(r+\hat{\theta}^2)T}+4x_0e^{3(r+2\hat{\theta}^2)T}} \right]^{\frac{1}{2}}$,
and the optimal state process for the dual problem is given by
\begin{equation}\label{non_hara_optimal_Y}
\hat{Y}(t)=\frac{1}{\sqrt{2x_0}}\left[ e^{(r+\hat{\theta}^2)T}+\sqrt{e^{2(r+\hat{\theta}^2)T}+4x_0e^{3(r+2\hat{\theta}^2)T}} \right]^{\frac{1}{2}}e^{(r-\frac{\hat{\theta}^2}{2})t+\hat{\theta}W(t)}.
\end{equation}
\end{description}
Solving the adjoint BSDE, we have
\begin{align*}
\hat{p}_2(t)\hat{Y}(t)&=E\left[ \hat{Y}(T)^{-3}+\hat{Y}(T)^{-1}|\mathcal{F}_t \right]\\
&= \hat{y}^{-3}e^{-3(r-\frac{\hat{\theta}^2}{2})T}e^{-3\hat{\theta}W(t)}e^{\frac{9}{2}\hat{\theta}^2(T-t)}+\hat{y}^{-1}e^{-(r-\frac{\hat{\theta}^2}{2})T}e^{-\hat{\theta}W(t)}e^{\frac{1}{2}\hat{\theta}^2(T-t)}
\end{align*}
Substituting \eqref{non_hara_optimal_Y} back into the above equation and rearranging, we have
\begin{equation}\label{non_hara_p_2}
\hat{p}_2(t)=\hat{y}^{-4}e^{-3(r+2\hat{\theta}^2)T}e^{-rt-4\hat{\theta}^2t-4\hat{\theta}W(t)}+\hat{y}^{-1}e^{-rT}e^{-rt-2\hat{\theta}W(t)}
\end{equation}
Applying Ito's formula, we have
\begin{equation*}
d\hat{p}_2(t)=\left[ -r\hat{p}_2(t)+4a_1\hat{\theta}^2S_1(t)+2a_2\hat{\theta}^2S_2(t) \right]dt-\left(4a_1\hat{\theta}S_1(t)+2a_2\hat{\theta}S_2(t)\right)dW(t),
\end{equation*}
where $ a_1= \hat{y}^{-4}e^{-3(r+\hat{\theta}^2)T},\ a_2=\hat{y}^{-1}e^{-rT} , S_1(t)=e^{-rt-4\hat{\theta}^2t-4\hat{\theta}W(t)}$ and $S_2(t)=e^{-rt-2\hat{\theta}W(t)}$ for $ t\in [0,T] $. We have
\begin{equation}\label{non_hara_q_2}
\hat{q}_2(t)=-4a_1\hat{\theta} S_1(t)-2a_2\hat{\theta} S_2(t), \ t\in[0,T].
\end{equation}
Finally, according to Theorem \ref{thm_dual_primal}, we can construct the optimal solution of the primal problem explicitly from optimal solution to the dual problem as
\begin{equation*}
\left\{
\begin{array}{l}
\hat{\pi}(t)=[\sigma^\intercal]^{-1}\hat{q}_2(t)\hat{p}_2^{-1}(t),\vspace{1mm}\\
X^{\hat{\pi}}(t)=\hat{p}_2(t)=\hat{y}^{-4}e^{-3(r+\hat{\theta}^2)T}e^{-rt-4\hat{\theta}^2t-4\hat{\theta}W(t)}+\hat{y}^{-1}e^{-rT}e^{-rt-2\hat{\theta}W(t)}.
\end{array}
\right.
\end{equation*}
\begin{remark}
Suppose that after attaining $ \hat{y} $ and $ v $, we try to recover the optimal solution to the primal problem directly. By duality relationship between the primal and dual value functions \cite[see Theorem 2.6]{bianzheng:JEDC}, we have
\begin{equation*}
u(t,x)=v(t,\hat{y}(x))+v_y(t,\hat{y}(x))\hat{y}(x)=\frac{2}{3}\left( \hat{y}(x)^{-1}e^{(r+\hat{\theta}^2)t}+2x\hat{y}(x) \right).
\end{equation*}
Hence, to get $ (\hat{\pi},X^{\hat{\pi}}) $, we would need to solve the following optimization problem on the Hamiltonian function:
\begin{equation*}
\hat{\pi}(t)=\argmin_{\pi\in K} \left[ \left( r(t)+\pi'\sigma(t)\theta(t) \right)u_x(t,x) + \frac{1}{2}tr\left( \sigma\sigma' u_{xx}(t,x) \right)\right].
\end{equation*}
and substituting the above back to the SDE \eqref{primal_state_process}, which appears to be highly complicated equation to solve. However, in the approach we proposed, the optimal adjoint processes of the dual problem can be written out explicitly as conditional expectations of the dual state process. The optimal solution to the primal problem can be constructed explicitly thanks to the dynamic relationship as stated in Theorem \ref{thm_dual_primal}.
\end{remark}
\section{Conclusions}
In this paper, we study constrained utility maximization problem following the convex duality approach. After formulating the primal and dual problems, we construct the necessary and sufficient conditions for both the primal and dual problems in terms of FBSDEs plus additional conditions. Such formulation then allows us to establish an explicit connection between primal and dual optimal solutions in a dynamic fashion. Finally we solve three constrained utility maximization problems using the dynamic convex duality approach we proposed above.


\end{document}